\documentclass[11pt]{amsart}
\usepackage{amsmath,amsthm,amssymb,xypic,array}
\usepackage[T1]{fontenc}
\usepackage{amsfonts}
\usepackage{amsmath}
\usepackage{amssymb}
\usepackage{amsthm}
\usepackage{setspace}
\usepackage[cp850]{inputenc}
\usepackage{hyperref}
\usepackage{graphicx}
\usepackage{fancyhdr}
\usepackage{mathdots}
\usepackage{tikz}
\usepackage{booktabs}
\usepackage{longtable}
\usepackage{geometry}

\usetikzlibrary{trees}

\providecommand{\U}[1]{\protect\rule{.1in}{.1in}}
\providecommand{\U}[1]{\protect\rule{.1in}{.1in}}
\providecommand{\U}[1]{\protect\rule{.1in}{.1in}}
\providecommand{\U}[1]{\protect\rule{.1in}{.1in}}
\providecommand{\U}[1]{\protect\rule{.1in}{.1in}}
\input{xy}
\xyoption{all}
\theoremstyle{theorem}

\newtheorem{Theorem}{Theorem}[section]

\newtheorem{Lemma}[Theorem]{Lemma}

\theoremstyle{definition}

\newtheorem{Remark}[Theorem]{Remark}

\newtheorem{Example}[Theorem]{Example}

\numberwithin{equation}{section}

\newcommand{\arXiv}[1]{\href{http://arxiv.org/abs/#1}{arXiv:#1}}

\newcommand{\rank}{\operatorname{rk}}
\newcommand{\brank}{\operatorname{\underline{rk}}}
\newcommand{\lrank}{\operatorname{rank}}
\newcommand{\Hom}{\operatorname{Hom}}

\newcommand{\lker}{\operatorname{Lker}}

\renewcommand{\a}{{\`a}}

\renewcommand{\sec}{\mathbb{S}ec}

\def\leq{\leqslant}
\def\geq{\geqslant}

\def\bibaut#1{{\sc #1}}

\def\phi{\varphi}
\def\ro[#1]{{\textcolor{red}{#1}}}

\renewcommand{\a}{\`a }

\begin{document}

\title{On the Rank of $n\times n$ Matrix Multiplication}

\author[Alex Massarenti]{Alex Massarenti}
\address{\sc Alex Massarenti\\
SISSA\\
via Bonomea 265\\
34136 Trieste\\ Italy}
\email{alex.massarenti@sissa.it}

\author[Emanuele Raviolo]{Emanuele Raviolo}
\address{\sc Emanuele Raviolo\\
Universit\a di Pavia\\
via Ferrata 1\\
27100 Pavia\\ Italy}
\email{emanuele.raviolo@unipv.it}

\date{\today}


\maketitle

\begin{abstract}
For every $p\leq n$ positive integer we obtain the lower bound $(3-\frac{1}{p+1})n^2-\big(2\binom{2p}{p+1}-\binom{2p-2}{p-1}+2\big)n$ for the rank of the $n\times n$ matrix multiplication. This bound improves the previous one $(3-\frac{1}{p+1})n^2-\big(1+2p\binom{2p}{p}\big)n$ due to Landsberg \cite{L}. Furthermore our bound improves the classic bound $\frac{5}{2}n^2-3n$, due to Bl\"aser \cite{B}, for every $n\geq 132$. Finally, for $p = 2$, with a sligtly different strategy we menage to obtain the lower bound $\frac{8}{3}n^2-7n$ which improves Bl\"aser's bound for any $n\geq 24$.
\end{abstract}

\tableofcontents

\section*{Introduction}
The multiplication of two matrices is one of the most important operations in mathematics and applied sciences. To determine the complexity of matrix multiplication is a major open question in algebraic complexity theory.\\
Recall that the matrix multiplication $M_{n,l,m}$ is defined as the bilinear map
\[
\begin{array}{cccc}
   M_{n,l,m}: & \textrm{Mat}_{n\times l}(\mathbb{C})\times\textrm{Mat}_{l\times m}(\mathbb{C}) &\to&  \textrm{Mat}_{n\times m}(\mathbb{C})\\
    &  (X,Y) & \mapsto & XY,
\end{array}
\]
where $\textrm{Mat}_{n\times l}(\mathbb{C})$ is the vector space of $n\times l$ complex matrices. A measure of the complexity of matrix multiplication, and of tensors in general, is the \textit{rank}. For the bilinear map $M_{n,l,m}$ this is the smallest natural number $r$ such that there exist $a_{1},...,a_{r}\in \textrm{Mat}_{n\times l}(\mathbb{C})^{*}$, $b_{1},...,b_{r}\in \textrm{Mat}_{l\times m}(\mathbb{C})^{*}$ and $c_{1},...,c_{r}\in \textrm{Mat}_{n\times m}(\mathbb{C})$ decomposing $M_{n,l,m}(X,Y)$ as
$$M_{n,l,m}(X,Y) = \sum_{i=1}^{r}a_{i}(X)b_{i}(X)c_{i}$$
for any $X\in \textrm{Mat}_{n\times l}(\mathbb{C})$ and $Y\in\textrm{Mat}_{l\times m}(\mathbb{C})$.\\
In the case of square matrices the standard algorithm gives an expression of the form $M_{n,n,n}(X,Y) = \sum_{i=1}^{n^{3}}a_{i}(X)b_{i}(X)c_{i}$. However \textit{V. Strassen} showed that that such algorithm is not optimal \cite{S1969}. In this paper we are concerned with lower bounds on the rank of matrix multiplication. The first lower bound $\frac{3}{2}n^2$ was proved by \textit{V. Strassen} \cite{S} and then improved by \textit{M. Bl\"aser} \cite{B}, who found the lower bound $\frac{5}{2}n^2-3n$.\\ 
Recently \textit{J.M. Landsberg} \cite{L}, building on work with \textit{G. Ottaviani} \cite{LO}, proved the new lower bound 
$\rank(M_{n,n,n})\geq (3-\frac{1}{p+1})n^2-\big(1+2p\binom{2p}{p}\big)n$, for every $p\leq n$, which improves Bl\"aser's bound for every $n\geq 85$.
In this paper, following Landsberg's strategy, we improve his bound for the rank of matrix multiplication.
\\
Our main result is the following.
\begin{Theorem}\label{main}
Let $p\leq n$ be a positive natural number. Then
\begin{equation}
\label{mp}
\rank(M_{n,n,m})\geq (1+\frac{p}{p+1})nm+n^2-\big(2\binom{2p}{p+1}-\binom{2p-2}{p-1}+2\big)n.
\end{equation}
When $n = m$ we obtain
\begin{equation}
\label{np}
\rank(M_{n,n,n})\geq (3-\frac{1}{p+1})n^2-\big(2\binom{2p}{p+1}-\binom{2p-2}{p-1}+2\big)n.
\end{equation}
\end{Theorem}
For example, when $p=3$, the bound (\ref{np}) becomes $\frac{11}{4}n^2-26n$, which improves Bl\"aser's one, $\frac{5}{2}n^2-3n$, for $n\geq132$.\\
Our strategy is the following. We prove Lemma \ref{key}, which is basically the improved version of \cite[Lemma 2.0.6]{L} applied to our case, using the classical identities for determinants of Lemma \ref{bm} and Lemma \ref{detsum}. The basic idea is to lower the degree of the equations that give the lower bound for border rank for matrix multiplication. Then we exploit this lower degree as Bl\"aser and Landsberg did.\\ 
%The strategy of the proof is the following. Given three vector spaces $L = \mathbb{C}^{n}, M = \mathbb{C}^{m}$ and $N = \mathbb{C}^{n}$ consider $A = N^{*}\otimes L, B = L^{*}\otimes M, C = N\otimes M^{*}$, and the matrix multiplication map
%$$M_{n,n,m}: (N^{*}\otimes L)\times (L^{*}\otimes M)\rightarrow N^{*}\otimes M.$$
%First decompose the matrix multiplication tensor $M_{n,n,m}$ as $M_{n,n,m} = \phi_{1}+\phi_{2}$ such that $\rank(\phi_{1}) = n^{2}$. The $n^{2}$ elements appearing in $\phi_{1}$ give a basis of $A^{*}$. Then we show there exists a subset of $n^{2}-(2p+3)n$ elements of these basis annihilating a maximal rank element $\alpha^{0}$, and some $\alpha^{1},...,\alpha^{2p}$, such that the matrix multiplication restricted to $A^{'} = \left\langle\alpha^{0},...,\alpha^{2p}\right\rangle$ has maximal rank. This is the aim of the Key Lemma \ref{key}. If $\psi_{1}$ is the sum of all monomials in $\phi_{1}$ whose terms in $A^{*}$ annihilate $\alpha^{0},...,\alpha^{2p}$ the number $n^{2}-(2p+3)n$ gives a lower bound on $\rank(\psi_{1})$. Furthermore, since the matrix multiplication restricted to $A^{'}$ has non-zero determinant, we get a lower bound on $\rank(\psi_{2})$, where $\psi_{2} = \phi_{1}-\psi_{1}+\phi_{2}$.\\ 
The paper is organized as follows. In Section \ref{pre} we give the basic definitions and explain the geometric meanings of the notions of rank and border rank in terms of secant varieties of Segre varieties. Section \ref{lo} is devoted to the Landsberg-Ottaviani equations \cite{LO}; we present them as rephrased in \cite{L}. Finally in Section \ref{kl} we prove the Key Lemma  and prove Theorem \ref{main}. In Remark \ref{p2} we consider the case $p=2$ obtaining a bound which improves Bl\"aser's one for every $n\geq24.$
\section{Preliminaries and Notation}\label{pre}
Let $V,W$ be two complex vector spaces of dimension $n$ and $m$. The contraction morphism
\[
\begin{array}{ccc}
    V^{*}\otimes W & \rightarrow &\Hom(V,W)      \\
      T=\sum_{i,j}f_{i}\otimes w_{j} & \mapsto & L_{T}
\end{array}
,\] 
where $L_{T}(v) = \sum_{i,j}f_{i}(v)w_{j}$, defines an isomorphism between $V^{*}\otimes W$ and the space of linear maps from $V$ to $W$.\\ 
Then, given three vector spaces $A,B,C$ of dimension $a,b$ and $c$, we can identify $A^{*}\otimes B$ with the space of linear maps $A\rightarrow B$, and $A^{*}\otimes B^{*}\otimes C$ with the space of bilinear maps $A\times B\rightarrow C$. Let $T:A^{*}\times B^{*}\rightarrow C$ be a bilinear map. Then $T$ induces a linear map $A^{*}\otimes B^{*}\rightarrow C$ and may also be interpreted as:
\begin{itemize}
\item[-] an element of $(A^{*}\otimes B^{*})^{*}\otimes C = A\otimes B\otimes C$,
\item[-] a linear map $A^{*}\rightarrow B\otimes C$.
\end{itemize}

%\subsubsection*{Secant Varieties}
%Let $X\subset\P^N$ be an irreducible and reduced non degenerate variety,
%$$\Gamma_{k}(X)\subset X\times ...\times X\times\G(k-1,N),$$
%the reduced closure of the graph of
%$$X\times ...\times X \dasharrow \G(k-1,N),$$
%taking $k$ general points to their linear span $\langle x_1,...,x_{k}\rangle$. Observe that $\Gamma_k(X)$ is irreducible and reduced of dimension $kn$. Let $\pi_2:\Gamma_k(X)\to\G(k-1,N)$ be the natural projection. Denote by
%$$\mathcal{S}_k(X):=\pi_2(\Gamma_k(X))\subset\G(k-1,N).$$
%Again $\mathcal{S}_k(X)$ is irreducible and reduced of dimension $kn$. Finally let 
%\[
%\begin{tikzpicture}
%\def\x{1.5}
%\def\y{-1.2}
%\node (A0_1) at (1*\x, 0*\y) {$\mathcal{I}_k=\{(x,\Lambda) \: | \: x\in \Lambda\}\subset\P^N\times\G(k-1,N)$};
%\node (A1_0) at (0*\x, 1*\y) {$\P^N$};
%\node (A1_2) at (2*\x, 1*\y) {$\G(k-1,N)$};
%\path (A0_1) edge [->] node [auto] {$\scriptstyle{\psi_{k}}$} (A1_2);
%\path (A0_1) edge [->] node [auto,swap] {$\scriptstyle{\pi_{k}}$} (A1_0);
%\end{tikzpicture}
%\]
%with natural projections $\pi_k$ and $\psi_k$ onto the factors.

%\begin{Definition} Let $X\subset\P^N$ be an irreducible and reduced, non degenerate variety. The {\it abstract $k$-Secant variety} is the irreducible and reduced variety
%$$\Sec_{k}(X):=(\psi_k)^{-1}(\mathcal{S}_k(X))\subset \mathcal{I}_k.$$ 
%While the {\it $k$-Secant variety} is
%$$\sec_{k}(X):=\pi_k(Sec_{k}(X))\subset\P^N.$$
%\end{Definition}

%Finally, we denote by $\sec_{k}(X)^{o}$ the open subset of $\sec_{k}(X)$ of points lying on the span of $k$ points of $X$ in linear general position.
\subsubsection*{Segre varieties and their secant varieties}
Let $A$, $B$ and $C$ be complex vector spaces.
The three factor Segre map is defined as 
\[
\begin{array}{cccc}
   \sigma_{1,1,1}:\mathbb{P}(A)\times\mathbb{P}(B)\times\mathbb{P}(C) &\rightarrow&\mathbb{P}(A\otimes B\otimes C)\\
    ([a],[b],[c]) & \mapsto & [a\otimes b\otimes c],
\end{array}
\]
where $[a]$ denotes the class in $\mathbb{P}(A)$ of the vector $a\in A$.
The notation $\sigma_{1,1,1}$ is justified by the fact that the Segre map is induced by the line bundle $\mathcal{O}(1,1,1)$ on 
$\mathbb{P}(A)\times\mathbb{P}(B)\times\mathbb{P}(C)$.
The two factor Segre map $$\sigma_{1,1}:\mathbb{P}(B)\times\mathbb{P}(C)\rightarrow\mathbb{P}(B\otimes C)$$
is defined in a similar way.
The Segre varieties are defined as the images of the Segre maps: 
$\Sigma_{1,1,1} = \sigma_{1,1,1}(\mathbb{P}(A)\times\mathbb{P}(B)\times\mathbb{P}(C))$,
$\Sigma_{1,1}=\sigma_{1,1}(\mathbb{P}(B)\times\mathbb{P}(C))$.
For each integer $r\geq0$ we define the open secant variety and the secant variety of $\Sigma_{1,1,1}$ respectively as 
\[
\mathbb{S}ec_r(\Sigma_{1,1,1})^o=\bigcup_{x_1,\dots,x_{r+1}\in\:\Sigma_{1,1,1}}\langle x_1,\dots,x_{r+1}\rangle, \    \   \     \ \mathbb{S}ec_{r}(\Sigma_{1,1,1})=\overline{\mathbb{S}ec_r(\Sigma_{1,1,1})^o}.
\]
In the above formulas $\langle x_1,\dots,x_{r+1}\rangle$ denotes the linear space generated by the points $x_i$ and $\mathbb{S}ec_{r}(\Sigma_{1,1,1})$ is the closure of $\mathbb{S}ec_r(\Sigma_{1,1,1})^o$ with respect to the Zariski topology. Let us notice that with the above definition $\mathbb{S}ec_0(\Sigma_{1,1,1})=\Sigma_{1,1,1}$.
\subsubsection*{Rank and border rank of a bilinear map}
The \textit{rank} of a bilinear map $T:A^{*}\times B^{*}\rightarrow C$ is the smallest natural number $r := \rank(T)\in\mathbb{N}$ such that there exist $a_{1},...,a_{r}\in A$, $b_{1},...,b_{r}\in B$ and $c_{1},...,c_{r}\in C$ decomposing $T(\alpha,\beta)$ as
$$T(\alpha,\beta) = \sum_{i=1}^{r}a_{i}(\alpha)b_{i}(\beta)c_{i}$$
for any $\alpha\in A^{*}$ and $\beta\in B^{*}$. The number $\rank(T)$ has also two additional interpretations.
\begin{itemize}
\item[-] Considering $T$ as an element of $A\otimes B\otimes C$ the rank $r$ is the smallest number of rank one tensors in $A\otimes B\otimes C$ needed to span a linear space containing the point $T$. Equivalently, $\rank(T)$ is the smallest number of points $t_{1},...,t_{r}\in\Sigma_{1,1,1}$ such that $[T]\in \left\langle t_{1},...,t_{r}\right\rangle$. In the language of secant varieties this means that $[T]\in\sec_{r-1}(\Sigma_{1,1,1})^{o}$ but $[T]\notin\sec_{r-2}(\Sigma_{1,1,1})^{o}$.
\item[-] Similarly, if we consider $T$ as a linear map $A^{*}\rightarrow B\otimes C$ then $\rank(T)$ is the smallest number of rank one tensors in $B\otimes C$ need to span a linear space containing the linear space $T(A^{*})$. As before we have a geometric counterpart. 
In fact $\rank(T)$ is the smallest number of points $t_{1},...,t_{r}\in\Sigma_{1,1}$ such that $\mathbb{P}(T(A^{*}))\subseteq \left\langle t_{1},...,t_{r}\right\rangle$.
\end{itemize}
The \textit{border rank} of a bilinear map  $T:A^{*}\times B^{*}\rightarrow C$ is the smallest natural number $r:=\underline{\textrm{rk}}(T)$ such that $T$ is the limit of bilinear maps of rank $r$ but is not a limit of tensors of rank $s$ for any $s< r$. There is a geometric interpretation also for this notion: $T$ has border rank $r$ if $[T]\in\sec_{r-1}(\Sigma_{1,1,1})$ but $[T]\notin\sec_{r-2}(\Sigma_{1,1,1})$. Clearly $\rank(T)\geq\underline{\textrm{rk}}(T)$.
\subsubsection*{Matrix multiplication} Now, let us consider a special tensor. Given three vector spaces $L = \mathbb{C}^{l}, M = \mathbb{C}^{m}$ and $N = \mathbb{C}^{n}$ we 
define $A=N\otimes L^*$, $B=L\otimes M^*$ and $C=N^*\otimes M$.
We have a matrix multiplication map
$$M_{n,l,m}: A^*\times B^*\rightarrow C$$
As a tensor $M_{n,l,m} = Id_{N}\otimes Id_{M}\otimes Id_{L}\in (N^*\otimes L)\otimes (L\otimes M^{*})\otimes (N^{*}\otimes M)=A\otimes B\otimes C$, where $Id_{N}\in N^{*}\otimes N$ is the identity map. If $n = l$ the choice of a linear map $\alpha^{0}:N\rightarrow L$ of maximal rank allows us to identify $N\cong L$. Then the multiplication map $M_{n,n,m}\in (N\otimes N^{*})\otimes (N\otimes M^{*})\otimes (N^{*}\otimes M)$ induces a linear map $N^{*}\otimes N\rightarrow (N^{*}\otimes M)\otimes (N^{*}\otimes M)^{*}$ which is an inclusion of Lie algebras 
$$M_A:\mathfrak{gl}(N)\rightarrow\mathfrak{gl}(B),$$
where $\mathfrak{gl}(N)\cong N^{*}\otimes N$ is the algebra of linear endomorphisms of $N$.
In particular, the rank of the commutator $[M_{A}(\alpha^1),M_{A}(\alpha^2)]$ of $nm\times nm$ matrices is equal to $m$ times the rank of the commutator $[\alpha^1,\alpha^2]$ of
$n\times n$ matrices.
This equality reflects a general philosophy, that is to translate expressions in commutators of $\mathfrak{gl}_{n^{2}}$ into expressions in commutators in $\mathfrak{gl}_{n}$.
\subsubsection*{Matrix equalities}\label{me}
The following lemmas are classical in linear algebra. However, for completeness, we give a proof. 
\begin{Lemma}\label{bm}
The determinant of a $2\times 2$ block matrix is given by
\[
\det
\begin{pmatrix}
X & Y\\
Z & W\\
\end{pmatrix}
=\det(X)\det(W-ZX^{-1}Y),
\]
where $X$ is an invertible $n\times n$ matrix, $Y$ is a $n\times m$ matrix, $Z$ is a $m\times n$ matrix, and $W$ is a $m\times m$ matrix.
\end{Lemma}
\begin{proof}
The statement follows from the equality
\[
\begin{pmatrix}
X & Y\\
Z & W
\end{pmatrix}
\begin{pmatrix}
-X^{-1}Y  & Id_{n}\\
Id_{m} & 0 
\end{pmatrix}
=
\begin{pmatrix}
0 & X\\
W-ZX^{-1}Y & Z
\end{pmatrix}
.\]
\end{proof}

\begin{Lemma}\label{detsum}
Let $A$ be an $n\times n$ invertible matrix and $U,V$ any $n\times m$ matrices. Then
$$\det(A+UV^t)=\det(A)\det(Id+V^t A^{-1}U),$$
where $V^t$ is the transpose of $V$.
\end{Lemma}
\begin{proof}
It follows from the equality
\[
\begin{pmatrix}
A & 0\\
V^t & Id
\end{pmatrix}
\begin{pmatrix}
Id  & -A^{-1}U\\
0 & Id+V^{t}A^{-1}U 
\end{pmatrix}
\begin{pmatrix}
Id & 0\\
-V^{t} & Id 
\end{pmatrix}
=
\begin{pmatrix}
A+UV^{t} & -U\\
0 & Id
\end{pmatrix}
.\]
\end{proof}

\section{Landsberg - Ottaviani equations}\label{lo}
In \cite{LO} \textit{J.M. Landsberg} and \textit{G. Ottaviani} generalized Strassen's equations as introduced by \textit{V. Strassen} in \cite{S}. We follow the exposition of \cite[Section 2]{L}.\\
Let $T\in A\otimes B\otimes C$ be a tensor, and assume $b = c$. Let us consider $T$ as a linear map $A^{*}\rightarrow B\otimes C$, and assume that there exists $\alpha\in A^{*}$ such that $T(\alpha):B^{*}\rightarrow C$ is of maximal rank $b$. Via $T(\alpha)$ we can identify $B\cong C$, and consider $T(A^{*})\subseteq B^{*}\otimes B$ as a subspace of the space of linear endomorphisms of $B$.\\
In \cite{S} \textit{Strassen} considered the case $a = 3$. Let $\alpha^{0},\alpha^{1},\alpha^{2}$ be a basis of $A^{*}$. Assume that $T(\alpha^{0})$ has maximal rank and that $T(\alpha^{1}),T(\alpha^{2})$ are diagonalizable, commuting endomorphisms. Then $T(\alpha^{1}),T(\alpha^{2})$ are simultaneously diagonalizable and 
it is not difficult to prove that in this case $\textrm{rk}(T)=b$.
In general, $T(\alpha^{1}),T(\alpha^{2})$ are not commuting. 
The idea of Strassen was to consider their commutator $[T(\alpha^{1}),T(\alpha^{2})]$ to obtain results on the border rank of $T$.
In fact, Strassen proved that, if $T(\alpha^0)$ is of maximal rank, then $\brank(T)\geq b+\lrank[T(\alpha^{1}),T(\alpha^{2})]/2$ and $\brank(T)=b$ if and only if $[T(\alpha^{1}),T(\alpha^{2})]=0$.\\
Now let us consider the case $a = 3, b = c$. Fix a basis $a_{0},a_{1},a_{2}$ of a $A$, and let $a^0,a^1,a^2$ be the dual basis of $A^*$. 
Choose bases of $B$ and $C$, so that elements of $B\otimes C$ can be written as matrices. Then we can write $T = a_{0}\otimes X_{0}-a_{1}\otimes X_{1}+a_{2}\otimes X_{2}$, where the $X_{i}$ are $b\times b$ matrices. Consider $T\otimes Id_{A}\in A\otimes B\otimes C\otimes A^{*}\otimes A = A^{*}\otimes B\otimes A\otimes A\otimes C$, 
$$T\otimes Id_{A} = (a_{0}\otimes X_{0}-a_{1}\otimes X_{1}+a_{2}\otimes X_{2})\otimes (a^{0}\otimes a_{0}+a^{1}\otimes a_{1}+a^{2}\otimes a_{2})$$
and its skew-symmetrization in the $A$ factor $T_{A}^{1}\in A^{*}\otimes B\otimes \bigwedge^{2}A\otimes C$, 
given by
$$T_A^1=a^{1}X_{0}(a_{0}\wedge a_{1})+a^{2}X_{0}(a_{0}\wedge a_{2})-a^{0}X_{1}(a_{1}\wedge a_{0})-a^{2}X_{1}(a_{1}\wedge a_{2})+a^{0}X_{2}(a_{2}\wedge a_{0})+a^{1}X_{2}(a_{2}\wedge a_{1})$$
where $a^{i}X_{j}(a_{j}\wedge a_{i}):= a^{i}\otimes X_{j}\otimes (a_{j}\wedge a_{i})$. It can also be considered as a linear map $$T_{A}^{1}:A\otimes B^{*}\rightarrow\bigwedge^{2}A\otimes C.$$
In the basis $a_{0},a_{1},a_{2}$ of $A$ and $a_{0}\wedge a_{1},a_{0}\wedge a_{2},a_{1}\wedge a_{2}$ of $\bigwedge^{2}A$ the matrix of $T_{A}^{1}$ is the following
\[
Mat(T_{A}^{1}) = 
\begin{pmatrix}
X_{1} & -X_{0} & 0\\
-X_{2} & 0 & X_{0}\\
0 & -X_{2} & -X_{1}
\end{pmatrix}
\]
Assume $X_{0}$ is invertible and change bases such that it is the identity matrix. By Lemma \ref{bm}, on the matrix obtained by reversing the order of the rows of $Mat(T_{A}^{1})$, with 
\[
X = 
\begin{pmatrix}
0 & X_{0}\\
X_{0} & 0
\end{pmatrix},\:
Y = 
\begin{pmatrix}
X_{1}\\
-X_{2}
\end{pmatrix},\:
Z = 
\begin{pmatrix}
-X_{1} & -X_{2}
\end{pmatrix},\:
W = 0
\]
we get 
$$\det(Mat(T_{A}^{1})) = \det(X_{1}X_{2}-X_{2}X_{1}) = \det([X_{1},X_{2}]).$$
Now we want to generalize this construction as done in \cite{LO}. We consider the case $a = 2p+1$, $T\otimes Id_{\bigwedge^{p}A}\in A\otimes B\otimes C\otimes\bigwedge^{p}A^{*}\otimes\bigwedge^{p}A = (\bigwedge^{p}A^{*}\otimes B)\otimes (\bigwedge^{p+1}A\otimes C)$, and its skew-symmetrization
$$T_{A}^{p}:\bigwedge^{p}A\otimes B^{*}\rightarrow\bigwedge^{p+1}A\otimes C.$$
Note that $\dim(\bigwedge^{p}A\otimes B^{*}) = \dim(\bigwedge^{p+1}A\otimes C) = \binom{2p+1}{p}b$.
After choosing a basis $a_{0},...,a_{2p}$ of $A$ we can write $T = \sum_{i=0}^{2p}(-1)^{i}a_{i}\otimes X_{i}$. 
To perform our computations in the proof of Lemma \ref{kl}, it is more convenient to consider the operator  $(T_{A}^{p})^*$, the transpose of $T_{A}^{p}$.
The matrix associated to $(T_{A}^{p})^*$ with respect the basis $a_{0}\wedge...\wedge a_{p-1},..., a_{p+1}\wedge...\wedge a_{2p}$ of $\bigwedge^{p}A$, and $a_{0}\wedge...\wedge a_{p},...,a_{p}\wedge...\wedge a_{2p}$ of $\bigwedge^{p+1}A$ is of the form
\begin{equation}\label{exp}
Mat((T_{A}^{p})^*) = 
\begin{pmatrix}
Q & 0\\
R & \overline{Q}
\end{pmatrix}
\end{equation}
where the matrix is blocked $(\binom{2p}{p+1}b,\binom{2p}{p}b)\times(\binom{2p}{p+1}b,\binom{2p}{p}b)$, the lower left block is given by
\[
R = 
\begin{pmatrix}
X_{0} & \hdots & 0\\
\vdots & \ddots & \vdots\\
0 & \hdots & X_{0}
\end{pmatrix}
\]
and $Q$ is a matrix having blocks $X_{1},...,X_{2p}$ and zero.
\\ The matrix is related to $Q$ in the following way. Write $Q=(Q_{i,j})$, where the $Q_{i,j}$ are  the $n\times n$ blocks of $Q$ and let $Q_{(k)}=(Q_{k,1},\dots,Q_{k,\binom{2p}{p}})$
be the $k$-th block-row of $Q$.
Then $\overline{Q}$ is the matrix whose $l$-th block-column is $Q^{(l)}=(Q_{\binom{2p}{p},\binom{2p}{p+1}-l+1},\dots,Q_{1,\binom{2p}{p+1}-l+1})$, with the convention that if $Q_{i,j}=X_{h}$, $h$ odd, then the block is multiplied by $-1$. 
We derive below the expression (\ref{exp}) in the case $p=2$; the general case can be developed similarly, see \cite[Section 4]{L}.
\begin{Example}\label{matp2}
In the case $p = 2$ we have 
$$
\begin{array}{ll}
T_{A}^{2} = & (a^{1}\wedge a^{2})X_{0}(a_{0}\wedge a_{1}\wedge a_{2})+(a^{1}\wedge a^{3})X_{0}(a_{0}\wedge a_{1}\wedge a_{3})+(a^{1}\wedge a^{4})X_{0}(a_{0}\wedge a_{1}\wedge a_{4})+\\
 & (a^{2}\wedge a^{3})X_{0}(a_{0}\wedge a_{2}\wedge a_{3})+(a^{2}\wedge a^{4})X_{0}(a_{0}\wedge a_{2}\wedge a_{4})+(a^{3}\wedge a^{4})X_{0}(a_{0}\wedge a_{3}\wedge a_{4})-\\
 & (a^{0}\wedge a^{2})X_{1}(a_{1}\wedge a_{0}\wedge a_{2})-(a^{0}\wedge a^{3})X_{1}(a_{1}\wedge a_{0}\wedge a_{3})-(a^{0}\wedge a^{4})X_{1}(a_{1}\wedge a_{0}\wedge a_{4})-\\
 & (a^{2}\wedge a^{3})X_{1}(a_{1}\wedge a_{2}\wedge a_{3})-(a^{2}\wedge a^{4})X_{1}(a_{1}\wedge a_{2}\wedge a_{4})-(a^{3}\wedge a^{4})X_{1}(a_{1}\wedge a_{3}\wedge a_{4})+\\
 & (a^{0}\wedge a^{1})X_{2}(a_{2}\wedge a_{0}\wedge a_{1})+(a^{0}\wedge a^{3})X_{2}(a_{2}\wedge a_{0}\wedge a_{3})+(a^{0}\wedge a^{4})X_{2}(a_{2}\wedge a_{0}\wedge a_{4})+\\
 & (a^{1}\wedge a^{3})X_{2}(a_{2}\wedge a_{1}\wedge a_{3})+(a^{1}\wedge a^{4})X_{2}(a_{2}\wedge a_{1}\wedge a_{4})+(a^{3}\wedge a^{4})X_{2}(a_{2}\wedge a_{3}\wedge a_{4})-\\
 & (a^{0}\wedge a^{1})X_{3}(a_{3}\wedge a_{0}\wedge a_{1})-(a^{0}\wedge a^{2})X_{3}(a_{3}\wedge a_{0}\wedge a_{2})-(a^{0}\wedge a^{4})X_{3}(a_{3}\wedge a_{0}\wedge a_{4})-\\
 & (a^{1}\wedge a^{2})X_{3}(a_{3}\wedge a_{1}\wedge a_{2})-(a^{1}\wedge a^{4})X_{3}(a_{3}\wedge a_{1}\wedge a_{4})-(a^{2}\wedge a^{4})X_{3}(a_{3}\wedge a_{2}\wedge a_{4})+\\
 & (a^{0}\wedge a^{1})X_{4}(a_{4}\wedge a_{0}\wedge a_{1})+(a^{0}\wedge a^{2})X_{4}(a_{4}\wedge a_{0}\wedge a_{2})+(a^{0}\wedge a^{3})X_{4}(a_{4}\wedge a_{0}\wedge a_{3})+\\
  & (a^{1}\wedge a^{2})X_{4}(a_{4}\wedge a_{1}\wedge a_{2})+(a^{1}\wedge a^{3})X_{4}(a_{4}\wedge a_{1}\wedge a_{3})+(a^{2}\wedge a^{3})X_{4}(a_{4}\wedge a_{2}\wedge a_{3})
\end{array}
$$
and the matrix of $T_{A}^{2}*$ is
\[
Mat(T_{A}^{2}) = 
\begin{pmatrix}
X_{2} & -X_{3} & X_{4} & 0 & 0 & 0 & 0 & 0 & 0 & 0\\
X_{1} & 0 & 0 & -X_{3} & X_{4} & 0 & 0 & 0 & 0 & 0\\
0 & X_{1} & 0 & -X_{2} & 0 & X_{4} & 0 & 0 & 0 & 0\\
0 & 0 & X_{1} & 0 & -X_{2} & X_{3} & 0 & 0 & 0 & 0\\
X_{0} & 0 & 0 & 0 & 0 & 0 & -X_{3} & X_{4} & 0 & 0\\
0 & X_{0} & 0 & 0 & 0 & 0 & -X_{2} & 0 & X_{4} & 0\\
0 & 0 & X_{0} & 0 & 0 & 0 & 0 & -X_{2} & X_{3} & 0\\
0 & 0 & 0 & X_{0} & 0 & 0 & -X_{1} & 0 & 0 & X_{4}\\
0 & 0 & 0 & 0 & X_{0} & 0 & 0 & -X_{1} & 0 & X_{3}\\
0 & 0 & 0 & 0 & 0 & X_{0} & 0 & 0 & -X_{1} & X_{2}\\
\end{pmatrix}
\]
If $X_{0}$ is the identity by Lemma \ref{bm} on $R = Id, Q$ and $\overline{Q}$ the determinant of $Mat(T_{A}^{2})$ is equal to the determinant of 
\begin{equation}\label{p=2}
\begin{pmatrix}
[X_{2},X_{3}] & -[X_{2},X_{4}] & [X_{3},X_{4}] & 0\\
[X_{1},X_{3}] & -[X_{1},X_{4}] & 0 & [X_{3},X_{4}]\\
[X_{1},X_{2}] & 0 & -[X_{1},X_{4}] & [X_{2},X_{4}]\\
0 & [X_{1},X_{2}] & -[X_{1},X_{3}] & [X_{2},X_{3}] 
\end{pmatrix}
\end{equation}
\end{Example}
In general the matrix $Q$ is as follows. Let us consider the entry $(i,j)$ of $Q$ corresponding to the basis vectors $a_{i_{1}}\wedge...\wedge a_{i_{p+1}}$ of $\bigwedge^{p+1}A$ and $a_{j_{1}}\wedge...\wedge a_{j_{p}}$ of $\bigwedge^{p}A$, and let $I=\{i_{1},...,i_{p+1}\}$, $J = \{j_{1},...,j_{p}\}$. Then
\begin{equation}\label{pgen}
Q_{i,j}=\left\{\begin{tabular}{ll}
$(-1)^{i+j}X_{k}$ & if \: $I,J$ \: differ by just one element \: $k$,\\
0 & otherwise.
\end{tabular}
\right.
\end{equation}
\begin{Remark}
\label{imp}
It follows from (\ref{pgen}) that  $Q\overline{Q}$ has always commutators as entries and a lower left block $\mathcal{X}_{1,2} = \textrm{diag}([X_{1},X_{2}])$ of size  $\binom{2p-2}{p-1}$. Furthermore on the diagonal of $Q\overline{Q}$ if there is an entry $[X_{i},X_{j}]$ then such entry appears at least twice. Finally on the diagonal all indices except $i = 1,2p$ appear if $p\geq 3$ and in the case $p = 2$ all indices appear as we can see from Example \ref{matp2}.
These features of  $Q\overline{Q}$ will be of central importance in the proof Lemma \ref{key}.
\end{Remark}
\begin{Example}\label{matp3}
Let us define $X_{ij}=[X_{i},X_{j}]$. 
Then for $p=3$ the matrix $Q\overline{Q}$ is
\[
\begin{pmatrix}
X_{34} & X_{35} & X_{36} & X_{45} & X_{46} & X_{56} & 0 & 0 & 0 & 0 & 0 & 0 & 0 & 0 & 0\\
X_{24} & X_{25} & X_{26} & 0 & 0 & 0 & X_{45} & X_{46} & X_{56} & 0 & 0 & 0 & 0 & 0 & 0\\
X_{23} & 0 & 0 & X_{25} & X_{26} & 0 & X_{34} & X_{35} & 0 & X_{56} & 0 & 0 & 0 & 0 & 0\\
0 & X_{23} & 0 & X_{24} & 0 & X_{26} &  X_{34} & 0 & X_{36} & X_{46} & 0 & 0 & 0 & 0 & 0\\
0 & 0 & X_{23} & 0 & X_{24} & X_{25} &  0 & X_{34} & X_{35} & X_{45} & 0 & 0 & 0 & 0 & 0\\
X_{14} & X_{15} & X_{16} & 0 & 0 & 0 & 0 & 0 & 0 & 0 & X_{45} & X_{46} & X_{56} & 0 & 0\\
X_{13} & 0 & 0 & X_{15} & X_{16} & 0 & 0 & 0 & 0 & 0 & X_{35} & X_{36} & 0 & X_{56} & 0\\
0 & X_{13} & 0 & X_{14} & 0 & X_{16} & 0 & 0 & 0 & 0 & X_{34} & 0 & X_{36} & X_{46} & 0\\
0 & 0 & X_{13} & 0 & X_{14} & X_{15} & 0 & 0 & 0 & 0 & 0 & X_{34} & X_{35} & X_{45} & 0\\
X_{12} & 0 & 0 & 0 & 0 & 0 & X_{15} & X_{16} & 0 & 0 & X_{25} & X_{26} & 0 & 0 & X_{56}\\
0 & X_{12} & 0 & 0 & 0 & 0 & X_{14} & 0 & X_{16} & 0 & X_{24} & 0 & X_{26} & 0 & X_{46}\\
0 & 0 & X_{12} & 0 & 0 & 0 & 0 & X_{14} & X_{15} & 0 & 0 & X_{24} & X_{25} & 0 & X_{45}\\
0 & 0 & 0 & X_{12} & 0 & 0 & X_{13} & 0 & 0 & X_{16} & X_{23} & 0 & 0 & X_{26} & X_{36}\\
0 & 0 & 0 & 0 & X_{12} & 0 & 0 & X_{13} & 0 & X_{15} & 0 & X_{23} & 0 & X_{25} & X_{35}\\
0 & 0 & 0 & 0 & 0 & X_{12} & 0 & 0 & X_{13} & X_{14} & 0 & 0 & X_{23} & X_{24} & X_{34}\\

\end{pmatrix}
\]
where we omit the signs for simplicity of notation. We suggest the reader to follow the proof of Lemma \ref{key} with the above matrix on hand.
\end{Example}

\section{Key Lemma}\label{kl}
We begin by recalling the following classical lemma which will be essential at every step of the proof of Lemma \ref{key}.
\begin{Lemma}\label{pol}\cite[Lemma 11.5.0.2]{Lb}
Let $V$ be a $n$-dimensional vector space and let $P\in S^{d}V^{*}\setminus\{0\}$ be a polynomial of degree $d\leq n-1$ on $V$. For any basis $\{v_{1},...,v_{n}\}$ of $V$ there exists a subset $\{v_{i_{1}},...,v_{i_{s}}\}$ of cardinality $s\leq d$ such that $P_{|\left\langle v_{i_{1}},...,v_{i_{s}}\right\rangle}$ is not identically zero.
\end{Lemma}
Lemma \ref{pol} says, for instance, that a quadric surface in $\mathbb{P}^{3}$ can not contain six lines whose pairwise intersections span $\mathbb{P}^{3}$. Note that as stated Lemma \ref{pol} is sharp in the sense that under the same hypothesis the bound $s\leq d$ can not be improved. For example the polynomial $P(x,y,z,w) = xy$ vanishes on the four points $[1:0:0:0],...,[0:0:0:1]\in\mathbb{P}^{3}$. 
\begin{Lemma}\label{key}
Let  $A=N^*\otimes L$, where $l=n$.  Given any basis of $A$, there exists a subset of at least
$h = n^{2}-(n(2\binom{2p}{p+1}-\binom{2p-2}{p-1}+2))$ basis vectors, and elements  $\alpha^0,\alpha^1,\dots, \alpha^{2p}$ of $A^*$, such that 
\begin{itemize}
\item[-] $\alpha^0$ is of maximal rank, and thus may be used to identify $L\simeq N$ and $A$ as a space of endomorphisms. 
(I.e. in bases $\alpha^0$ is the identity matrix.)
\item[-] Choosing a basis of $L$, so the $\alpha^j$ become
$n\times n$ matrices, the block matrix of (\ref{pgen}) whose blocks are the  $\alpha^i$ is such that $Q\overline{Q}$ has non-zero determinant,  and
\item[-] The subset of at least $h$ basis vectors annihilate $\alpha^0, \alpha^1,\dots, \alpha^{2p}$.
\end{itemize}
\end{Lemma}
\begin{proof}
Let $\mathcal{B}$ be a basis of $A$, and consider the polynomial $P_{0} = \det_{n}$. By Lemma \ref{pol} we get a subset $S_{0}$ of at most $n$ elements of $\mathcal{B}$ and $\alpha^{0}\in S_{0}$ with $\det_{n}(\alpha^{0})\neq 0$. Now, via the isomorphism $\alpha^{0}:L\rightarrow N$ we are allowed to identify $A = \mathfrak{gl}(L)$ as an algebra with identity element $\alpha^{0}$. So, from now on, we work with $\mathfrak{sl}(L) = \mathfrak{gl}(L)/\left\langle\alpha^{0}\right\rangle$ instead of $\mathfrak{gl}(L)$.\\
Let $v_{1,0},..., v_{2p,0}\in\mathfrak{sl}(L)$ be linearly independent and not equal to any of the given basis vectors, and let us work locally on an affine open neighborhood $\mathbb{V}\subset G(2p,\mathfrak{sl}(L))$ of $E_{0} = \left\langle v_{1,0},..., v_{2p,0}\right\rangle$. We extend $v_{1,0},..., v_{2p,0}$ to a basis $v_{1,0},...,v_{2p,0},w_{1},..., w_{n^{2}-2p-1}$ of $\mathfrak{sl}(L)$, and take local coordinates $(f^{\mu}_{s})$ with $1\leq s\leq 2p$, $1\leq\mu\leq n^{2}-2p-1$, on $V$, so that $v_{s} = v_{s,0}+\sum_{\mu=1}^{n^{2}-2p-1}f^{\mu}_{s}w_{\mu}$.\\
We denote $v_{i,j} = [v_{i},v_{j}]$ and let us consider the matrix $Q$ of \ref{pgen} whose entries are the $v_i$ and the matrix $M = Q\overline{Q}$. Let $A$ be the diagonal matrix constructed as follows:
$$
A_{i,i}=\left\{\begin{tabular}{ll}
$M_{i,i}$ & if \: $M_{i,i}\neq 0$,\\
$Id$ & if \: $M_{i,i} = 0$ \: or \: if $M_{i,i} = [v_{j},v_{k}]$ \: with $j\in\{1,2p\}$ \: or \: $k\in\{1,2p\}$.
\end{tabular}
\right.
$$
and let $U = M-A$. Let us stress that by Remark \ref{imp}, as soon as $p\geq 3$, the last condition is automatically satisfied because an the diagonal of $M$ there are not commutators $[v_{j},v_{k}]$ with $j = 1,2p$ or $j = 1,2p$. The polynomial $\det(M)$ is not identically zero on $G(2p,\mathfrak{sl}(L))$ by \cite{LO}, so it is not identically zero on $\mathbb{V}$. Furthermore by Remark \ref{imp} any nontrivial entry of $A$ appears at least twice.\\ 
So $P_{1} = \det(A)$ is a polynomial of degree at most $2n\binom{2p}{p+1}$ being $A$ a matrix of size $n\binom{2p}{p+1}$ with quadratics entries. However the reduced polynomial $\widetilde{P}_1$ induced by $P_{1}$ has degree at most $n\binom{2p}{p+1}$ because any nontrivial entry of $A$ appears at least twice. Applying Lemma \ref{pol} to $\widetilde{P}_{1}$ we find a subset $S_{1}$ of at most $n\binom{2p}{p+1}$ elements of our basis such that $\widetilde{P}_{1}$,
and hence $P_1$, is not identically zero on $\left\langle S_{1}\right\rangle$.\\
Now, we can write $M = A+U Id$, and by Lemma \ref{detsum} we have
$$\det(M) = \det(A)\det(Id+A^{-1}U).$$
Let us fix some particular value of the coordinates $f_s^\mu$ such that on the corresponding matrices $\overline{v}_{2},...,\overline{v}_{2p-1}$ the matrix $A$ is invertible. For these values the expression $\det(Id + A^{-1}U)$ makes sense. Furthermore the matrix $Id + A^{-1}U$ has the following block form
\[
Id + A^{-1}U=\begin{pmatrix}
Y & Z\\
\mathcal{X}_{1,2} & W
\end{pmatrix}
\]
where $\mathcal{X}_{1,2}$ is a diagonal matrix, with $[v_{1},v_{2}]$ on the diagonal, of size $n\binom{2p-2}{p-1}$ with linear entries because we fixed $\overline{v}_{2}$. 
Hence $P_{2}=\det(\mathcal{X}_{1,2})$ is a polynomial of degree $n\binom{2p-2}{p-1}$ whose reduced polynomial $\widetilde{P}_{2}=\det([X_{1},X_{2}])$
 has degree $n$. 
 By Lemma \ref{pol} we find a subset $S_{2}$ of at most $n$ elements of the basis $\mathcal{B}$ such that $\widetilde{P}_{2}$ and hence $P_{2}$ is 
 not identically zero on $\left\langle S_{2}\right\rangle$. We then fix some values of the coordinates $f^{\mu}_{s}$ in such a way that the corresponding matrix 
 $\overline{v}_{1}$ is such that $P_{2}$ is not zero.
 By Lemma \ref{bm} on $Id + A^{-1}U$ we get 
$$\det(Id + A^{-1}U) = \det(\mathcal{X}_{1,2})\det(Z-Y\mathcal{X}_{1,2}^{-1}W).$$ 
 Let us consider $P_{3} = \det(Z-Y\mathcal{X}_{1,2}^{-1}W)$.
 The blocks $Y,Z,W$ have linear entries because we already fixed $\overline{v}_{1},...,\overline{v}_{2p-1}$. 
 Moreover the product matrix $Y\mathcal{X}_{1,2}^{-1}W$ has linear entries as well because the blocks $[X_{j},X_{2p}]$ in $Y$ and $W$ never multiply each other.
 Hence $\det(Z-Y\mathcal{X}_{1,2}^{-1}W)$ has degree equal to the order of $Z-Y\mathcal{X}_{1,2}^{-1}W$, which is $n(\binom{2p}{p+1}-\binom{2p-2}{p-1})$.
 Again by Lemma \ref{pol} we find a subset $S_{3}$ of at most $n(\binom{2p}{p+1}-\binom{2p-2}{p-1})$ elements of the basis $\mathcal{B}$ such that $P_{3}$ is not identically zero on $\left\langle S_{3}\right\rangle$.\\
Summing up we found a subset $S$ of at most 
$$n+n\binom{2p}{p+1}+n+n(\binom{2p}{p+1}-\binom{2p-2}{p-1}) = n(2\binom{2p}{p+1}-\binom{2p-2}{p-1}+2)$$
elements of $\mathcal{B}$ such that $\det(M)$ is not identically zero on $\left\langle S\right\rangle$.
\end{proof}

%\begin{Remark}
%In \cite[Lemma 4.3]{L} the author proved the analogous statement for $n^{2}-(4p+1)n$.
%\end{Remark}

\begin{proof}[Proof of Theorem \ref{main}]
Let $\phi$ be a decomposition of the matrix multiplication tensor $M_{n,n,m}$ as sum of $r = \rank(M_{n,n,m})$ rank one tensors. Recall that the left kernel of a bilinear map $f:V\times U\rightarrow W$ is defined as 
$$\lker(f) = \{v\in V \: | \: f(v,u) = 0 \: \forall \: u\in U\}.$$ 
Since $\lker(M_{n,n,m}) = 0$, that is for any $\alpha\in A^{*}\setminus\{0\}$, there exists $\beta\in B^{*}$ such that $M_{n,n,m}(\alpha,\beta)\neq 0$ we can write $\phi = \phi_{1}+\phi_{2}$ with $\rank(\phi_{1}) = n^{2}$, $\rank(\phi_{2}) = r-n^{2}$ and $\lker(\phi_{1}) = 0$.\\
The $n^{2}$ elements of $A^{*}$ appearing in $\phi_{1}$ form a basis of $A^{*}$. By Lemma \ref{key} there exists a subset of $n^{2}-(n(2\binom{2p}{p+1}-\binom{2p-2}{p-1}+2))$ of them annihilating a maximal rank element $\alpha^{0}$ and some $\alpha^{1},...,\alpha^{2p}$ such that, choosing bases, the determinant of the matrix $([\alpha^{i},\alpha^{j}])$ is non-zero.\\
Let $\psi_{1}$ be the sum of all monomials in $\phi_{1}$ whose terms in $A^{*}$ annihilate $\alpha^{0},...,\alpha^{2p}$. By Lemma \ref{key} there exists at least $n^{2}-(n(2\binom{2p}{p+1}-\binom{2p-2}{p-1}+2))$ of them. Then $\rank(\psi_{1})\geq n^{2}-(n(2\binom{2p}{p+1}-\binom{2p-2}{p-1}+2))$. Furthermore consider $\psi_{2} = \phi_{1}-\psi_{1}+\phi_{2}$ so that $\phi = \psi_{1}+\psi_{2}$ and the terms appearing in $\psi_{2}$ does not annihilate $\alpha^{0},...,\alpha^{2p}$.\\
Let $A^{'} = \left\langle\alpha^{0},...,\alpha^{2p}\right\rangle\subseteq A^{*}$. Again by Lemma \ref{key} the determinant of the linear map $M_{n,n,m|A^{'}\otimes B^{*}\otimes C^{*}}:\bigwedge^{p}A^{'}\otimes B^{*}\rightarrow\bigwedge^{p+1}A^{'}\otimes C$ is non-zero. Then $\brank(\phi_{2})\geq nm\frac{2p+1}{p+1} = \dim(\bigwedge^{p}A^{'}\otimes B^{*})$. We conclude that
$$\rank(\phi) = \rank(\phi_{1})+\rank(\phi_{2})\geq n^{2}-(n(2\binom{2p}{p+1}-\binom{2p-2}{p-1}+2))+nm\frac{2p+1}{p+1}.$$ 
In particular, if $m = n$ we get
$$\rank(\phi) = \rank(\phi_{1})+\rank(\phi_{2})\geq (3-\frac{1}{p+1})n^{2}-(n(2\binom{2p}{p+1}-\binom{2p-2}{p-1}+2)).$$ 
\end{proof}
\begin{Remark}\label{p2}
When $p=2$ the bound \ref{np} can be improved because the matrix $M$ has a particular shape. In fact it has the same determinant of the following matrix,
 which we also call $M$ with a slight abuse of notation: 
\[M=
\begin{pmatrix}
0 & X_{1,2} & X_{1,3} & X_{1,4} \\
-X_{1,2} & 0 & X_{2,3} & X_{2,4} \\
-X_{1,3} & -X_{2,3} & 0 & X_{3,4} & \\
-X_{1,4} & -X_{2,4} & -X_{3,4}  & 0\\
 \end{pmatrix}
\]
where $X_{i,j}$ denotes the commutator matrix $[X_{i},X_{j}] = X_{i}X_{j}-X_{j}X_{i}$.
Let the matrices $v_s$ be defined as in the proof of Lemma \ref{pol} and define 
$$
A_{1,2}=\begin{pmatrix}
0 & v_{1,2}\\
-v_{1,2} & 0 
\end{pmatrix}
$$
Let $A$ be the following diagonal block matrix
\[
A=\textrm{diag}(A_{1,2},Id_{2n\times 2n})
\]
which is a squared matrix of order $8n$ and write $M=A+U$, with $U=M-A$.
We can now reason as in the proof of Lemma \ref{pol}, with the only exception that in this case the polynomial $\det(Id+A^{-1}U)$ will have degree 4n, instead of the order of the matrix
$Id+A^{-1}U$, which is $8n$.
We get the bound
\[ 
\rank(M_{n,n,n})\geq \frac{8}{3}n^2-7n
\]
which improves Bl\"aser's bound for every $n\geq 24$.\\
By Example \ref{matp3} we know the matrix $M$ for $p = 3$ as well. Following the proof of Lemma \ref{key} we see that $\deg(\widetilde{P}_{1}) = 6n$. Note that the computation for general $p$ of Lemma \ref{key} just implies that $\deg(\widetilde{P}_{1}) \leq 15n$. So, in this case we obtain
\[ 
\rank(M_{n,n,n})\geq \frac{11}{4}n^2-17n
\]
which improves the bound $\frac{8}{3}n^2-7n$ for any $n\geq 120$.
\end{Remark}
\vspace{5mm}
\subsubsection*{Acknowledgments}
Both the authors were introduced to this topic by \textit{J.M. Landsberg} during the Summer School \textit{"Tensors: Waring problems and Geometric Complexity Theory"} held in Cortona in July 2012. We would like to thank \textit{J.M. Landsberg} and \textit{M. Mella} for their beautiful lectures and all the participants for the stimulating atmosphere. We thank primarily \textit{J.M. Landsberg} for suggesting us the problem, for his interest and his generous hints.


\begin{thebibliography}{999}
\bibitem[B]{B}\bibaut{M. Bl\"aser}, \textit{A $\frac{5}{2}n^2$ lower bound for the rank of $n\times n$-matrix multiplication over arbitrary fields}, 440th Annual Symposium on Foundations of Computer Science (New York, 1999), IEEE Computer Soc, Los Alamitos, CA, 1999, pp. 45-50, MR MR1916183.
\bibitem[L]{L} \bibaut{J.M. Landsberg}, \textit{New lower bounds for the rank of matrix multiplication}, \arXiv{1206.1530}.
\bibitem[L1]{Lb} \bibaut{J.M. Landsberg}, \textit{Tensors: geometry and applications}, Graduate Studies in Mathematics, vol. 128, American
Mathematical Society, Providence, RI, 2012, MR 2865915.
\bibitem[LO]{LO} \bibaut{J.M. Landsberg, G. Ottaviani}, \textit{New lower bounds for the border rank of matrix multiplication}, \arXiv{1112.6007}.
%\bibitem[P]{P} \bibaut{P.D. Powell}, \textit{Calculating determinants of block matrices}, \arXiv{1112.4379}.
\bibitem[S]{S1969} \bibaut{V. Strassen},  \textit{Gaussian Elimination is not Optimal}, Numer. Math. 13, p. 354-356, 1969.
\bibitem[S1]{S} \bibaut{V. Strassen}, \textit{Rank and optimal computation of generic tensors}, Linear Algebra Appl. 52/53(1983), 645-685. MR 85b:15039.
\end{thebibliography}
\end{document}